\documentclass[a4paper,12pt,leqno]{article}
\usepackage{upref,amsmath,amsthm,amsfonts,amssymb}
\usepackage[titletoc,toc,title]{appendix}
\makeatletter
\def\lastpage@putlabel{}
\makeatother
\numberwithin{equation}{section}
\newcommand\boldB{{\boldsymbol{B}}}
\newcommand\boldL{{\boldsymbol{L}}}

\newcommand\boldWhat{\hat{\boldsymbol{W}}}
\newcommand{\calA}{\mathcal{A}}
\newcommand{\calC}{\mathcal{C}}
\newcommand{\calD}{\mathcal{D}}
\newcommand{\calE}{\mathcal{E}}
\newcommand{\calR}{\mathcal{R}}
\newcommand{\Comp}{\mathbb{C}}
\newcommand{\der}{\partial}
\DeclareMathOperator{\diag}{diag}
\renewcommand{\epsilon}{\varepsilon}
\newcommand{\Integer}{\mathbb{Z}}

\newcommand\Mat{\operatorname{Mat}\nolimits}

\newcommand\ord{\operatorname{ord}\nolimits}
\newcommand\semicomm[1]{\underset{#1}\sim}

\newcommand\val{\operatorname{val}}

\newcommand\lemref[1]{Lemma~\ref{#1}}
\newcommand\thmref[1]{Theorem~\ref{#1}}
\newcommand\propref[1]{Proposition~\ref{#1}}
\newcommand\corref[1]{Corollary~\ref{#1}}
\newcommand\remref[1]{Remark~\ref{#1}}
\newcommand\secref[1]{\S\ref{#1}}

\newtheorem{thm}{Theorem}[section]
\newtheorem{prop}[thm]{Proposition}
\newtheorem{lem}[thm]{Lemma}
\newtheorem{cor}[thm]{Corollary}
\theoremstyle{definition}
\newtheorem{defn}[thm]{Definition}

\theoremstyle{remark}
\newtheorem{rem}[thm]{Remark}
\newcommand\commentout[1]{}
\begin{document}
\title%
{General Zakharov-Shabat equations without Lax operators}
\author{Masatoshi~Noumi\thanks{Masatoshi NOUMI:
Department of Mathematics, Rikkyo University, Toshima-Ku, Tokyo
171-8501, Japan; passed away on 20 November 2024.}
\and
Takashi~Takebe\thanks{Takashi TAKEBE:
Beijing Institute of Mathematical Sciences and Applications,
No.~544, Hefangkou Village, Huaibei Town, Huairou District,
Beijing, 101408, People's Republic of China;
e-mail: takebe@bimsa.cn }
}
\date{}



\maketitle
\begin{abstract}
 The operators in the Zakharov-Shabat equations of integrable
 hierarchies are usually defined from the Lax operators. In this article
 it is shown that the Zakharov-Shabat equations themselves recover the
 Lax operators under suitable change of independent variables in the
 case of the KP hierarchy and the modified KP hierarchy (in the matrix
 formulation).
\end{abstract}


\section*{Introduction}
\label{sec:intro}

The KP hierarchy, which Sato introduced in \cite{sat:81},
(see also \cite{sat:89-1}, \cite{sat:89-2}, \cite{sat-nou:84},
\cite{sat-sat}), is a system of differential equations for a
microdifferential (or pseudodifferential) operator
\begin{equation*}
 \begin{split}
    L
    &=
    \der + u_{-1}(x,t) \der^{-1} + u_{-2}(x,t) \der^{-2} + \dotsb
\\
    &=
    \sum_{j=0}^\infty u_{1-j}(x,t) \der^{1-j}, 
    \qquad
    u_1=1,\ u_0=0,\ \der=\der_x,
\end{split}
\end{equation*}
with independent variables $x$ and $t=(t_1,t_2,\dotsc)$. The system has
several formulations, among which we consider the {\em Lax
representation},
\[
    \frac{\der L}{\der t_n}=[B_n, L], \quad
    B_n:= (L^n)_+,\  n=1,2,\dotsc,
\]
($[P,Q]=PQ-QP$; see \eqref{def:diff-part} for the notation $P_+$), and
the {\em Zakharov-Shabat representation},
\[
    \left[
     \frac{\der}{\der t_m} - B_m,
     \frac{\der}{\der t_n} - B_n
    \right]
    =
    \frac{\der B_m}{\der t_n} - \frac{\der B_n}{\der t_m}
    + [B_m, B_n] = 0,
    \quad m,n=1,2,\dotsc,
\]
which can be interpreted as a {\em zero-curvature condition}. 
(Other formulations are the Sato equations for the wave operator or the
Hirota bilinear equations for the tau-function, upon which we do not
touch in this article.)

It is easy to prove that the Lax representation and the Zakharov-Shabat
representation are equivalent. However, in some problems, the
Zakharov-Shabat equations {\em without a Lax operator} appear. (See, for
example, \cite{kri:92}.) A natural question is: Do those equations
without a Lax operator define a system equivalent to the KP hierarchy?

In this article we prove that general Zakharov-Shabat equations for
differential operators $B_n$ of order $n$ ($n=1,2,\dotsc$) imply the
existence of the Lax operator which satisfies the KP hierarchy after
suitable coordinate changes. We also prove a similar statement for the
modified KP hierarchy in the matrix formulation. In both cases we first
construct the $L$-operators satisfying the Lax equations by rather
complicated induction, using ``semicommutativity'' guaranteed by the
Zakharov-Shabat equations. Then we prove that those $L$-operators
reproduce the given $B$-operators by usual procedures for the KP
hierarchy or the modified KP hierarchy with coordinate changes.

The case for the KP hierarchy was published in \cite{take:92} with
numerous typos and insufficient arguments. We correct those typos and
fill gaps in the proofs in the present article. But the gaps in the
proof of \cite{take:92} for the Toda lattice hierarchy are still
open. The half of that proof for the Toda lattice hierarchy is used
with corrections in the proof for the modified KP hierarchy.

\bigskip
This article is divided into two parts: the first part is on the KP
hierarchy and the second part is on the matrix modified KP hierarchy. In
\secref{subsec:review:kp} we review the definition of the KP hierarchy
and its basic properties, in particular, the equivalence of the Lax
representation and the Zakharov-Shabat representation in the usual
setting. In \secref{subsec:semicomm} we introduce the notion of
``semicommutativity'' which will be used in the inductive
construction of the $L$-operator in
\secref{subsec:construction-L:kp}. This $L$ is a solution of the KP
hierarchy with suitably changed independent variables, as shown in
\secref{subsec:recovery-kp}.

The structure of \secref{sec:general-zs->lax:mKP} for the matrix
modified KP hierarchy is almost the same as that of
\secref{sec:general-zs->lax:kp}. First we review the matrix formulation
of the modified KP hierarchy introduced by Dickey \cite{dic:99} in
\secref{subsec:review:mkp}. (See also \cite{take:02}.) The construction
of the $L$-operator is done by induction in
\secref{subsec:construction-L:mkp}. The final step to the mKP hierarchy
in \secref{subsec:recovery-mkp} is a little bit different from the
approach for the KP hierarchy in \secref{subsec:recovery-kp}.

\section{Zakharov-Shabat equations for the KP hierarchy}
\label{sec:general-zs->lax:kp}

\subsection{Review of the KP hierarchy}
\label{subsec:review:kp}

First let us briefly review facts about microdifferential operators.
(For details of the theory of microdifferential operators we refer to
\cite{scha} and \cite{kas:83}.)

We consider microdifferential operators over $\calR = \calC[[x]]$ or
$\calR =\calC((x))$, where the ring of constants $\calC$ is $\calA[[t]]$
for a ($\Comp$-)algebra $\calA$. Here we mean by a {\em microdifferential
operator} a formal series of the form
\begin{equation}
    P = \sum_{k\in\Integer} a_k \der^k, \quad 
    (a_k \in \calR,\ a_k = 0 \text{ for }k\gg 0).
\label{def:microdiff}
\end{equation}
Product of two microdifferential operators
\begin{equation}
    P = \sum_{k \in\Integer} a_k \der^k, \qquad
    Q = \sum_{k \in\Integer} b_k \der^k,
\label{microdiff-P,Q}
\end{equation}
is defined by the extended Leibniz rule,
\begin{equation}
    PQ =
    \sum_{k \in\Integer} 
    \left(\sum_{l+m-r=k, r\geq 0} \binom{l}{r} a_l b_m^{(r)} \right)
    \der^k.
\label{mult:microdiff}
\end{equation}
The set of microdifferential operators forms an associative
$\calR$-algebra, which we denote by $\calE$. Its subset of differential
operators is denoted by $\calD$:
\begin{equation}
    \calD:=\left\{P = \sum_{k\geq0} a_k \der^k \in \calE\right\}.
\label{def:diff-op}
\end{equation}
The integer
\begin{equation}
    \ord (P) := \max \{k\in\Integer\,|\, a_k \neq 0 \, \}
\label{def:order-E}
\end{equation}
is called the {\em order} of the microdifferential operator $P$. (We
define $\ord(0)=-\infty$.) The order defines the filtration
$(\calE(n))_{n\in\Integer}$ of $\calE$: 
\begin{equation}
    \calE(n) :=
    \{ P \in \calE \,|\, \ord(P) \leqq n \}.
\label{def:order-filtration}
\end{equation}
Any microdifferential operator $P$ \eqref{def:microdiff} is decomposed
uniquely into the {\em differential operator part} $P_+ \in \calD$
and the {\em negative order part} $P_- \in \calE(-1)$ as follows:
\begin{equation}
\begin{gathered}
    P = P_+ + P_-,\\
    P_+ :=  \sum_{k\geq 0} a_k \der^k,\quad
    P_- :=  \sum_{k\leq -1} a_k \der^k.
\end{gathered}
\label{def:diff-part}
\end{equation}
Namely $\calE = \calD \oplus \calE(-1)$. We call an operator $P$ monic,
if it has the form $P=\der^n + (\text{terms of order}<n)$, and denote
the set of monic operators by $\calE^{\text{monic}}$.

\bigskip
Using the above terminology, we can formulate the KP hierarchy as
follows. The {\em KP hierarchy} is a system of differential equations
for a first-order microdifferential operator
\begin{equation}
    L
    = \der + u_{-1} \der^{-1} + u_{-2} \der^{-2} + \cdots
    = \sum_{i=0}^\infty u_{1-i} \der^{1-i},
\label{lax-op:kp}
\end{equation}
with $u_1 = 1, u_0 = 0, u_i = u_i(t,x) \in \calR$, in other words, a
system of differential equations for a series of functions $\{u_i\}_i$. 

The {\em Lax representation} of the KP hierarchy is
\begin{equation}
    \frac{\der L}{\der t} = [B_n, L], \quad
    B_n := (L^n)_+, \qquad n=1,2,\dotsc.
\label{lax-eq:kp}
\end{equation}
Another description of the KP hierarchy is the {\em Zakharov-Shabat
representation} or the {\em zero-curvature equations},
\begin{equation}
    \left[
     \frac{\der}{\der t_m}-B_m, \frac{\der}{\der t_n}-B_n
    \right]
     =
     \frac{\der B_n}{\der t_m} - \frac{\der B_m}{\der t_n} + [B_m,B_n]
     = 0,
     \quad
     m,n=1,2,\dotsc.
\label{zs:kp}
\end{equation}
It is not difficult to show that, if $L$ satisfies the Lax equations
\eqref{lax-eq:kp}, then it gives a solution of the Zakharov-Shabat
equations \eqref{zs:kp}, and that, if $B_n$'s defined from $L$ satisfy
the Zakharov-Shabat equations, then $L$ is a solution of the Lax
equations. In fact, if $L$ satisfies the Lax equations, its powers
satisfy 
\begin{equation}
    \frac{\der L^m}{\der t_n} = [B_n, L^m],
\label{[Bn,Lm]}
\end{equation}
for all $m$ and $n$. Subtracting \eqref{[Bn,Lm]} from the same
equation with $m$ and $n$ interchanged, we have
\begin{equation}
 \begin{split}
    \frac{\der L^m}{\der t_n} - \frac{\der L^n}{\der t_m}
    &= [B_n, L^m] - [B_m, L^n].
\\
    &= [B_n, B_m] + [(L^m)_-, (L^n)_-].
\end{split}
\label{[Bn,Bm]}
\end{equation}
The differential operator part of this equation gives the
Zakharov-Shabat equation \eqref{zs:kp}. Conversely, if $L$ defines the
$B$-operators satisfying the Zakharov-Shabat equations \eqref{zs:kp} for
any $m$ and $n$, then we have
\begin{equation}
    \frac{\der L^m}{\der t_n} - [B_n, L^m]
    =
    \frac{\der B_n}{\der t_m} + \frac{\der (L^m)_-}{\der t_n}
    - [B_n, (L^m)_-],
\label{[der-B,L]+zs}
\end{equation}
whose right-hand side is of order not more than $n-1$ and, in
particular, bounded for any $m$. If the Lax equation \eqref{lax-eq:kp}
does not hold and
\begin{equation}
    \frac{\der L}{\der t_n} - [B_n, L] 
    =
    a \der^M + (\text{lower order terms}), \qquad
    a\neq 0,
\label{[der-B,L]:tmp2}
\end{equation}
it is easy to see that the left-hand side of \eqref{[der-B,L]+zs} is
equal to
\[
 \begin{split}
    \frac{\der L^m}{\der t_n} - [B_n, L^m]
    &=
    \sum_{k=0}^{m-1} L^{k}
    \left(\frac{\der L}{\der t_n} - [B_n, L] \right)
    L^{m-k-1}
\\
    &= m a \der^{M+m-1} + (\text{lower order terms}),
 \end{split}
\]
the order of which is not bounded when $m$ increases. This is a
contradiction.  Hence the left-hand side of \eqref{[der-B,L]:tmp2}
should vanish and $L$ satisfies the Lax equation \eqref{lax-eq:kp}.

The above derivation of the Lax equations from the Zakharov-Shabat
representation heavily depends on the existence of the Lax operator
$L$. The goal of this article is to derive the Lax equations from the
Zakharov-Shabat equations without starting from the Lax operator.

\subsection{Semicommutative operators}
\label{subsec:semicomm}

In this subsection we prove several facts about ``almost'' commuting
microdifferential operators. The main tool here is the estimate of the
order of the commutator of microdifferential operators,
\begin{equation}
    \ord([A,B])\leqq \ord(A)+\ord(B)-1, \qquad
    A, B\in\calE,
\label{ord[A,B]}
\end{equation}
which is easily proved. 

\begin{prop}
\label{prop:semicommutative}
 Let $P=\der^m+p_{m-1}(x)\der^{m-1}+p_{m-2}(x)\der^{m-2}+\dotsb$ be a
 monic microdifferential operator $\in\calE(m)$ and $L$ its monic $m$-th
 root: $L=\der+u_0(x)+u_{-1}(x)\,\der^{-1}+\dotsb$, $L^m=P$. Then for
 any $Q\in\calE(n)$ the following three statements are equivalent
 ($r\geqq0$).

\begin{enumerate}
\renewcommand{\labelenumi}{(\roman{enumi})}
 \item $[P,Q]\in\calE(m+n-1-r)$.
 \item $[L,Q]\in\calE(n-r)$.
 \item In the expansion 
\begin{equation}
    Q=a_n(x)\, L^n + a_{n-1}(x)\, L^{n-1}+\dotsb,
\label{Q=sum(akLk)}
\end{equation}
       the first $r$ coefficients are constants: $\der a_i(x)=0$
       ($i=n,n-1,\dotsc,n-r+1$).
\end{enumerate}
\end{prop}

\begin{rem}
\label{rem:root-of-diff-op}
 The monic $m$-th root $L$ of $P$ always exists, as the condition
 $L^m=P$ gives relations of $u_i$'s and $p_j$'s as differential
 polynomials and by solving these relations recursively, we can express
 each $u_i$ as a differential polynomial of $p_j$'s.
 
 In particular, $u_0(x)=p_{m-1}(x)/m$. Hence, if $p_{m-1}(x)=0$ and $P$
 is of the form $\der^m + p_{m-2}(x)\der^{m-2}+\dotsb$, then $u_0(x)=0$
 and $L=\der + u_{-1}(x)\der^{-1}+\dotsb$.
\end{rem}

\begin{rem}
\label{rem:expansion-by-monic:kp}
 An expansion like \eqref{Q=sum(akLk)} always exists. In fact, if
 $Q=q_n(x)\der^n+(\text{terms of order }<n)$, we take $q_n(x)$ as
 $a_n(x)$. Then $Q':= Q-a_n(x) L^n$ is of order $<n$ and we can find
 $a_k$'s recursively.
\end{rem}

\begin{proof}[Proof of \propref{prop:semicommutative}]
 The proofs of (ii) from (iii) and of (i) from (ii) are
 straightforward. 

 To prove (iii) from (i), assume that $[P,Q]\in\calE(m+n-1-r)$. For a
 non-constant element $a\not\in\calC$, $\ord([P,a])=\ord(m(\der
 a)\der^{m-1}+ (\text{terms with orders}<m-1))=m-1$. Hence, if
 $a_{k+1},a_{k+2},\dotsc, a_n$ in the expansion \eqref{Q=sum(akLk)} are
 constants and $a_k\not\in\calC$, then $\ord([P,Q])=(m-1)+k$, as
\[
 \begin{split}
    [P,Q] &=
    \sum_{j\leq n} [P, a_j] L^j
\\
    &= [P,a_k] L^k + [P,a_{k-1}] L^{k-1} + \dotsb.
 \end{split}
\]
 Since $[P,Q]\in\calE(m+n-1-r)$, we have $m-1+k\leqq m+n-1-r$, i.e.,
 $k\leqq n-r$, which implies (iii).
\end{proof}

\begin{defn}
\label{def:r-commutative}
 Monic operators $P,Q \in\calE^{\text{monic}}$ satisfying the
 conditions in \propref{prop:semicommutative} are called {\em
 $r$-commutative} and we denote $P\semicomm{r}Q$. In general, if
 $P\semicomm{r}Q$ for some $r\geqq0$, we say that $P$ and $Q$ are
 semicommutative. 
\end{defn}

\begin{cor}
\label{cor:r-commutative=equiv}
 The relation $\semicomm{r}$ is an equivalence relation between monic
 operators. 
\end{cor}

\begin{proof}
 We have only to check that $P\semicomm{r}Q$ and $Q\semicomm{r}R$
 imply $P\semicomm{r}R$ for monic microdifferential operators $P$, $Q$
 and $R$.

 \propref{prop:semicommutative} says that there are expansions
\[
    P = \sum_{k\leq \ord(P)} p_k(x) L^k, \qquad
    R = \sum_{l\leq \ord(R)} r_l(x) L^l,
\]
 where $\der p_k(x) = \der r_l(x) = 0$ for $\ord(P)-r+1\leqq k \leqq
 \ord(P)$ and $\ord(R)-r+1\leqq l \leqq \ord(R)$. Therefore,
\[
 \begin{split}
    [P,R] &= 
    \left[
     \sum_{k\leq \ord(P)} p_k(x) L^k, \sum_{l\leq \ord(R)} r_l(x) L^l
    \right]
\\
    &=
    \left[
     \sum_{k\leq \ord(P)-r} p_k(x) L^k, \sum_{l\leq \ord(R)} r_l(x) L^l
    \right]
\\
    &\qquad +
    \left[
     \sum_{\ord(P)-r+1\leq k\leq \ord(P)} p_k(x) L^k, 
     \sum_{l\leq \ord(R)-r} r_l(x) L^l
    \right].
 \end{split}
\]
 Both commutators in the last expression have orders not more than
 $\ord(P)+\ord(R)-r-1$. Hence $P\semicomm{r}R$ by
 \propref{prop:semicommutative} (i).
\end{proof}

\begin{prop}
\label{prop:semicomm-Bn}
 Let $\{B_n\}_{n=1,2,\dotsc}$ be a sequence of monic differential
 operators of order $n$:
 $B_n=\der^{n}+b_{n,n-1}\der^{n-1}+b_{n,n-2}\der^{n-2}+\dotsb+b_{n,0}$.

 For any $r\geqq0$ the following three statements are equivalent:

\begin{enumerate}
\renewcommand{\labelenumi}{(\roman{enumi})}
 \item For any $n$ and $m$, if $n\leqq m$, then
       $[B_n,B_m]\in\calE(m-1-r)$, i.e., $B_n\underset{n+r}\sim B_m$.
 \item There exists a monic first-order microdifferential operator $L$
       such that $B_n\underset{n+r}\sim L$ for any $n$.
 \item There exists a monic first-order microdifferential operator $L$
       such that in the expansion $B_n= \sum_{k\leq n}a_{n,k} L^k$,
       $a_{n,n},\dotsc,a_{n,1-r}\in\calC$ for any $n$.
\end{enumerate}

 Moreover, if each $B_n$ is of the form $\der^n + b_{n,n-2}\der^{n-2} +
 \dotsb + b_{n,0}$, then $L$ is of the form $\der +
 u_{-1}\der^{-1}+\dotsb$. 
\end{prop}

\begin{proof}
 The equivalence of (ii) and (iii) is a consequence of
 \propref{prop:semicommutative}. Indeed, by taking $B_n$ as $Q$ and
 $n+r$ as $r$ in \propref{prop:semicommutative}, statements (ii) and
 (iii) there correspond to statements (ii) and (iii) in
 \propref{prop:semicomm-Bn} respectively. 

 If all $B_n$'s are expanded as in (iii) with constant
 $a_{n,n},\dotsc,a_{n,1-r}$, then the commutator of $B_n$ and $B_m$
 ($n\leqq m$) is
\[
 \begin{split}
    [B_n,B_m] &=
    \left[
      \sum_{k\leq n}a_{n,k} L^k, \sum_{l\leq m}a_{m,l} L^l
    \right]
\\
    &=
    \left[ \sum_{k\leq -r}a_{n,k} L^k, B_m \right]
    +
    \left[
      \sum_{1-r\leq k\leq n}a_{n,k} L^k, \sum_{l\leq -r}a_{m,l} L^l
    \right].
 \end{split}
\]
 The first commutator in the last line is an operator of order not
 more than $-r+m-1$, while the second commutator is an operator of order
 not more than $n-r+1$. Therefore $\ord([B_n,B_m])\leqq m-r-1$, which
 implies (i).

 Conversely, let us assume the statement (i). We denote the monic $n$-th
 root of $B_n$ by $L_n$:
\begin{equation}
    L_n = \der + u_{n,0}(x) + u_{n,-1} \der^{-1} + \dotsb,\qquad
    B_n = (L_n)^n.
\label{Ln=root-of-Bn}
\end{equation}
 Then, \propref{prop:semicommutative} implies $L_n\semicomm{n+r}L_m$ if
 $n\leqq m$ because of $B_n\semicomm{n+r}B_m$.

 We construct a sequence of monic first-order operators %
 $\{\tilde L_n\}_{n=1,2,\dotsc}$ which satisfies
\begin{equation}
    \tilde L_n \semicomm{n+r} L_n, \qquad
    \tilde L_{n} - \tilde L_{n-1} \in \calE(2-n-r).
\label{tildeLn}
\end{equation}
 Due to the second condition in \eqref{tildeLn} the sequence has a limit
 in the topology of $\calE$ by the order, $L:=\lim_{n\to\infty} \tilde
 L_n$. The operator $\tilde L_n$ differs from $L$ by an operator $\tilde
 L_n^c:=\tilde L_n-L$ of order not more than $1-n-r$ as $\tilde
 L_{k+1}-\tilde L_k\in\calE(1-k-r)\subset\calE(1-n-r)$ for $k\geqq
 n$. Hence,
\[
    [L,\tilde L_n] = [L,\tilde L_n^c]
    \in \calE(1+(1-n-r)-1) = \calE(1-n-r),
    \text{ i.e., } L\semicomm{n+r}\tilde L_n.
\]
 Since $\tilde L_n\semicomm{n+r}L_n$ by the first condition in
 \eqref{tildeLn}, we have $L\semicomm{n+r}L_n$ by
 \corref{cor:r-commutative=equiv}.

 Now \propref{prop:semicommutative} implies that
 $L\semicomm{n+r}B_n$, which is nothing but the statement (ii) of
 \propref{prop:semicomm-Bn}. 

\medskip
 The construction of $\tilde L_n$ is by induction. We start from $\tilde
 L_1=L_1$ and use the following lemma to find $\tilde L_n$ from $\tilde
 L_{n-1}$.

\begin{lem}
\label{lem:construction-tildeLn}
 Let $M$ and $N$ be monic first-order operators which are mutually
 $s$-commutative: $M\semicomm{s}N$. Then there exists a monic
 first-order operator $\tilde N$ such that 
\begin{equation}
    \tilde N \semicomm{s+1} N, \qquad
    \tilde N - M \in\calE(1-s).
\label{M,N->tildeN}
\end{equation}
\end{lem}

 Applying this lemma to $M=\tilde L_{n-1}$ and $N=L_n$ (in this case
 $s=n-1+r$), we obtain $\tilde L_n$ as $\tilde N$. (The condition
 $M\semicomm{s}N$, i.e., $\tilde L_{n-1}\semicomm{n-1+r}L_n$ follows
 from $L_{n-1}\semicomm{n-1+r}L_n$ which we showed after
 \eqref{Ln=root-of-Bn} and the induction hypothesis $\tilde
 L_{n-1}\semicomm{n-1+r}L_{n-1}$.) 

\begin{proof}[Proof of \lemref{lem:construction-tildeLn}]
 \propref{prop:semicommutative} implies that in the expansion
\[
    N = \sum_{n\leq1}a_n M^n = M + a_0 + a_{-1} M^{-1} + \dotsb,
\]
 the first $s$ coefficients $a_1=1$, $a_0, \dotsc, a_{2-s}$ are
 constant. Let us define $\tilde N$ by 
\[
 \begin{split}
    \tilde N &:= N - (a_0 + a_{-1} M^{-1} + \dotsb + a_{2-s}M^{2-s})
\\
    &=
    M + a_{1-s} M^{1-s} + a_{-s} M^{-s} + \dotsb
    \in M + \calE(1-s),
 \end{split}
\]
 which satisfies the second condition in \eqref{M,N->tildeN}. 

 Since $M\semicomm{s}N$, one has $M^n\semicomm{s}N$, i.e.,
 $[M^n,N]\in\calE(n+1-s)$ for any $n$ by
 \propref{prop:semicommutative}. Therefore $[M^n,N]\in\calE(-s)$ for
 $n\leqq-1$. Hence,
\[
    [\tilde N,N] = [N,N] - \sum_{n=2-s}^0 a_n[M^n,N]
    = - \sum_{n=2-s}^{-1} a_n[M^n,N] \in\calE(-s),
\]
 which is the first condition \eqref{M,N->tildeN}. 
\end{proof}

 The last statement of \propref{prop:semicomm-Bn} is due to
 \remref{rem:root-of-diff-op} and the above construction of $L$.
\end{proof}

\subsection{Construction of the Lax operator of the KP hierarchy}
\label{subsec:construction-L:kp}

In this subsection we assume the Zakharov-Shabat equations for the
sequence $\{B_n\}_{n=1,2,\dotsb}$ of monic differential operators and
construct a first order monic microdifferential operator which satisfies
the Lax equations with the given $B_n$'s.

\begin{prop}
\label{prop:general-zs->lax:kp}
 Let $\{B_n\}_{n=1,2,\dotsc}$ be a sequence of monic differential
 operators of order $n$: $\ord(B_n)=n$. If this sequence satisfies the
 system of Zakharov-Shabat equations,
\begin{equation}
    \left[
     \frac{\der}{\der t_n}-B_n, \frac{\der}{\der t_m}-B_m
    \right]
    =
    \frac{\der B_n}{\der t_m} - \frac{\der B_m}{\der t_n} +[B_n,B_m]
    = 0,
\label{general-zs:kp}
\end{equation}
 for $n,m = 1,2,\dotsc$, then there exists a monic first-order
 microdifferential operator $L=\der+u_0+u_{-1}\der^{-1}+\dotsb$ which
 satisfies the Lax equations
\begin{equation}
    \left[\frac{\der}{\der t_n}-B_n,L\right]
    =
    \frac{\der L}{\der t_n} - [B_n,L]
    = 0.
\label{general-lax:kp}
\end{equation}
 for $n=1,2,\dotsc$.

 Moreover, if each $B_n$ is of the form $\der^n + b_{n,n-2}\der^{n-2} +
 \dotsb + b_{n,0}$, then $u_0=0$ and $L$ is of the form $\der +
 u_{-1}\der^{-1}+\dotsb$.
\end{prop}

\begin{proof}
 Since $\der_{t_m}B_n - \der_{t_n} B_m\in\calE(m-1)$ for $n\leqq m$, the
 Zakharov-Shabat equation \eqref{general-zs:kp} implies
 $[B_n,B_m]\in\calE(m-1)$. \propref{prop:semicomm-Bn} ($r=0$) guarantees
 existence of the operator $L_0$ of the form
 $L_0=\der+u_{0,0}(x)+u_{0,-1}(x)\der^{-1}+\dotsb$ which satisfies
 $[B_n,L_0]\in\calE(0)$ for any $n$.

 Starting from $L_0$, we construct a sequence $\{L_r\}_{r=0,1,2\dotsc}$
 such that 
\begin{equation}
    \frac{\der L_r}{\der t_n} - [B_n, L_r] \in \calE(-r)
\label{lax-up-to-E(-r)}
\end{equation}
 for any $n$ and
\begin{equation}
    L_{r+1}-L_r \in \calE(-r).
\label{Lr+1-Lr-in-E(-r)}
\end{equation}
 It is obvious that $L_0$ satisfies \eqref{lax-up-to-E(-r)} ($r=0$) by
 construction. 

 Assume that $L_r$ satisfies \eqref{lax-up-to-E(-r)} for $r\geqq 0$. We
 construct $L_{r+1}$ from $L_r$ in the form
\begin{equation}
    L_{r+1} = L_r - f L_r^{-r},
\label{Lr+1}
\end{equation}
 where $\der f = 0$. As $L_{r+1}-L_r=-f L_r^{-r}\in\calE(-r)$, the
 condition \eqref{Lr+1-Lr-in-E(-r)} is satisfied.

 Let us define functions $g_n$ and $\varphi_{n,k}$ by
\begin{gather*}
    \left[\frac{\der}{\der t_n}-B_n, L_r\right]
    =
    \frac{\der L_r}{\der t_n} - [B_n,L_r] = g_n \der^{-r} + \dotsb,
\\
    B_n = \sum_{k\leq n} \varphi_{n,k}L_r^k, \qquad \varphi_{n,n}=1.
\end{gather*}
 Since $[B_n,L_r]\in\calE(0)$ by assumption \eqref{lax-up-to-E(-r)},
 $B_n\semicomm{n}L_r$. Hence it follows from \propref{prop:semicomm-Bn}
 (iii) that $\der \varphi_{n,k}=0$ for $k\geqq1$.

 Let us compute
\[
    R :=
    \frac{\der B_m}{\der t_n} - [B_n,B_m]
    =
    \frac{\der B_n}{\der t_m} \in \calE(n-1),
\]
 using the expansion of $B_m$. By the definition of $g_n$, we have
\begin{equation}
     \left[\frac{\der}{\der t_n}-B_n, L_r^k \right]
     =
     kg_n\der^{k-1-r} + \dotsb
\label{[d-Bn,Lrk]} 
\end{equation}
 for any $k\in\Integer$. Hence, as $\der\varphi_{m,k}=0$ for $k\geqq1$,
\[
 \begin{split}
    R &= 
    \left[
     \frac{\der}{\der t_n}-B_n, \sum_{k\leq m}\varphi_{m,k}L_r^k
    \right]
\\
    &=
    \sum_{k\leq m} 
    \left[
     \frac{\der}{\der t_n}-B_n, \varphi_{m,k}
    \right]L_r^k
    +
    \sum_{k\leq m} 
    \varphi_{m,k} (kg_n\der^{k-1-r} + \dotsb)
\\
    &=
    \sum_{k\leq m} \frac{\der\varphi_{m,k}}{\der t_n} L_r^k
    -
    \sum_{k\leq 0} [B_n, \varphi_{m,k}] L_r^k
    +
    \sum_{k\leq m} 
    \varphi_{m,k} (kg_n\der^{k-1-r} + \dotsb).
 \end{split}
\]
 Since $R$ and the second sum in the last expression are of order not
 more than $n-1$,
\[
    -\sum_{k\leq m} \frac{\der\varphi_{m,k}}{\der t_n} L_r^k
    \equiv
    m g_n \der^{m-1-r} + (\text{lower order terms})
    \mod\calE(n-1).
\]
 Therefore, if $m>n+r$, 
\begin{equation}
    \frac{\der \varphi_{m,k}}{\der t_n}
    =
    \begin{cases}
    0, &(k>m-1-r),\\
    -m g_n, &(k=m-1-r).
    \end{cases}
\label{dvarphimk/dtn}
\end{equation}
 Differentiating the case $k=m-1-r$ by $x$, we have
\begin{equation}
    -m\, \der g_n = \frac{\der}{\der t_n} \der\varphi_{m,m-1-r}.
\label{-mdgn:kp}
\end{equation}
 Let us take $m$ larger than $n+r+1$. Then $m>n+r+1>n+r$ and
 $m-1-r>n\geqq1$. Therefore by \eqref{-mdgn:kp} and
 $\der\varphi_{n,k}=0$ ($k\geqq 1$), we have $\der g_n=0$.

 To find $f$ in \eqref{Lr+1}, set
\begin{equation}
     f_m(t) := -\frac{1}{m}(\varphi_{m,m-1-r}(t)-\varphi_{m,m-1-r}(t=0)).
\label{correction-term:kp:fm}
\end{equation}
 Then, by \eqref{dvarphimk/dtn}, when $m>n+r+1$, we have
 $\der_{t_n}f_m=g_n$ which does not depend on $m$. This means that $f_m$
 has a limit $f:=\lim_{m\to\infty}f_m$ in $\Comp[[t_1,t_2,\dotsc]]$ with
 the topology by valuation ($\val t_n = n$). (Note that $f_m(t=0)=0$ for
 any $m$.) We have $\der_{t_n}f = g_n$ and, since $\varphi_{m,k}$
 ($k\geqq1$) is constant with respect $x$, $\der f =0$.

 If we set $L_{r+1}:=L_r - fL_r^{-r}$ as in \eqref{Lr+1}, $L_{r+1}$
 satisfies
\[
 \begin{split}
    \left[\frac{\der}{\der t_n}-B_n, L_{r+1}\right]
    &=
    \left[\frac{\der}{\der t_n}-B_n, L_{r}\right]
    -
    \left[\frac{\der}{\der t_n}-B_n, f L_{r}^{-r}\right]
\\
    &=
    \bigl(g_n \der^{-r} + (\text{terms of order\,}\leqq -r-1) \bigr)
\\
    &\qquad
    -\bigl(
      \frac{\der f}{\der t_n} L_r^{-r} 
      +
      f\left[\frac{\der}{\der t_n}-B_n, L_{r}^{-r}\right]
    \bigr)
\\
    &=
    f\left[\frac{\der}{\der t_n}-B_n, L_{r}^{-r}\right]
    + (\text{terms of order\,}\leqq -r-1),
 \end{split}
\]
 which is of order not more than $-r-1$ by \eqref{[d-Bn,Lrk]}. Hence, we
 have an operator $L_{r+1}$ satisfying the conditions
 \eqref{lax-up-to-E(-r)} for $r\mapsto r+1$.

 Thanks to the condition \eqref{Lr+1-Lr-in-E(-r)}, the sequence has a
 limit $L:=\lim_{r\to\infty}L_r$ in the topology by filtration of
 $\calE$. The limit $L$ satisfies the Lax equations
 \eqref{general-lax:kp} as the limit of \eqref{lax-up-to-E(-r)}.

 The final statement of \propref{prop:general-zs->lax:kp} follows from
 the final statement of \propref{prop:semicomm-Bn} and the construction
 of $L$.

 This completes the proof of \propref{prop:general-zs->lax:kp}.
\end{proof}

\subsection{Recovery of the KP hierarchy}
\label{subsec:recovery-kp}

Finally, we connect the above result with the KP hierarchy by means of
a coordinate change.

\begin{thm}
\label{thm:general-zs->lax:kp}
 Let $\{B_n\}_{n=1,2,\dotsb}$ be a sequence of monic differential
 operators of the form
\begin{equation}
    B_n = \der^n + b_{n,n-1}\der^{n-1} + \dotsb + b_{n,0},
\label{general-B-op:kp}
\end{equation}
 satisfying the system of Zakharov-Shabat equations
 \eqref{general-zs:kp} for $n,m=1,2,\dotsc$. Then there exists a monic
 first-order microdifferential operator,
\begin{equation}
    L=\der+u_0(x)+u_{-1}(x)\der^{-1}+\dotsb,
\label{general-L-op:kp}
\end{equation}
 a coordinate change,
\begin{equation}
    t_n \mapsto \tilde t_n = t_n + O(t_{n+1},t_{n+2},\dotsc)
\label{coord-change:general-zs->kp}
\end{equation}
 and a function $f_0=f_0(x,t_1,t_2,\dotsc)$ such that:
\begin{itemize}
 \item The operator $L$ satisfies Lax equations,
\begin{equation}
    \frac{\der L}{\der \tilde t_n}
    =
    [(L^n)_+, L],
\label{kp-h:general}
\end{equation}
 with respect to the new time variables $(\tilde t_n)_{n=1,2,\dotsc}$.  

 \item Each $B_n$ is expressed in terms of $L$ as
\begin{equation}
    B_n
    =
    \tilde B_n + \frac{\der\tilde t_{n-1}}{\der t_n} \tilde B_{n-1}
    + \dotsb +
    \frac{\der\tilde t_1}{\der t_n} \tilde B_1 +
    \frac{\der f_0}{\der t_n}, \quad
    \tilde B_n := e^{f_0}(L^n)_+ e^{-f_0}.
\label{Bn->tildeBn}
\end{equation}
\end{itemize}

 Moreover, if $b_{n,n-1}=0$ in \eqref{general-B-op:kp} for all $n$, then
 $u_0=0$, which means that $L$ is a Lax operator of the KP hierarchy
 with respect to the variables $(\tilde t_n)_{n=1,2,\dotsc}$. Moreover
 $\der f_0=0$, from which it follows that $\tilde B_n = (L^n)_+$.
\end{thm}

\begin{proof}
 Let us denote the operator $L$ obtained in
 \propref{prop:general-zs->lax:kp} by $\tilde L$. Expanding $B_n$ as a
 power series of $\tilde L$,
\[
    B_n = \sum_{k=0}^\infty \varphi_{n,k} {\tilde L}^k,
    \quad \varphi_{n,n}=1,
\]
 and taking the differential operator part, we obtain the expression
 of $B_n$ as follows:
\begin{equation}
    B_n = \sum_{k=0}^n \varphi_{n,k} \tilde B_k, \qquad
    \tilde B_n := ({\tilde L}^n)_+,
\label{Bn->tildeBn:temp}
\end{equation}
 Since $[B_n, \tilde L]=\der_{t_n}\tilde L \in\calE(0)$, $B_n$ and
 $\tilde L$ are $n$-commutative: $B_n\semicomm{n}\tilde L$. Therefore
 \propref{prop:semicommutative} (iii) implies
\begin{equation}
    \der \varphi_{n,k}=0 \quad \text{for $1\leqq k\leqq n$,}
\label{phink:const}
\end{equation}
 from which it follows that $[\varphi_{n,k},\tilde L]=0$ for
 $1\leqq k \leqq n$. Let us define an operator $B_n^c$ by
\begin{equation}
    B_n^c
    := B_n - \sum_{k=0}^n \varphi_{n,k} \tilde L^n
    = - \sum_{k=0}^n \varphi_{n,k} (\tilde L^n)_-.
\label{Bnc:kp}
\end{equation}
 Replacing $B_n$ in the Lax equation $[\der_{t_n}-B_n,\tilde L]=0$ by
 $B_n^c+\sum_{k=0}^n\varphi_{n,k}\tilde L^k$, we obtain
\[
   \left[\frac{\der}{\der t_n} - B_n^c-\varphi_{n,0}, \tilde L\right]
   = 0,
\]
 because of \eqref{phink:const}.
 Using these Lax equations and \eqref{phink:const}, we have
\begin{equation}
 \begin{split}
    0 &=
    \left[
     \frac{\der}{\der t_n} - B_n,
     \frac{\der}{\der t_m} - B_m
    \right]
\\
    &=
    \left[
     \frac{\der}{\der t_n} - B_n^c - \varphi_{n,0},
     \frac{\der}{\der t_m} - B_m^c - \varphi_{m,0}
    \right]
    +
    \sum_{k\geq1}
    \left(
     \frac{\der\varphi_{n,k}}{\der t_m}
     -
     \frac{\der\varphi_{m,k}}{\der t_n}
    \right) \tilde L^k
\\
    &=
    \sum_{k\geq0}
    \left(
     \frac{\der\varphi_{n,k}}{\der t_m}
     -
     \frac{\der\varphi_{m,k}}{\der t_n}
    \right) \tilde L^k
    + (\text{a negative order operator}).
 \end{split}
\label{sum(dphi/dt-dphi/dt)Lk:kp}
\end{equation}
 Here we define $\varphi_{n,k}=0$ for $k>n$. Because of this convention, 
\begin{equation}
     \frac{\der\varphi_{n,k}}{\der t_m}
     =
     \frac{\der\varphi_{m,k}}{\der t_n},
\label{dphink/dtm=dphimk/dtn:kp}
\end{equation}
 holds trivially for sufficiently large $k$. Starting from this large
 $k$ and proceeding downwards by \eqref{sum(dphi/dt-dphi/dt)Lk:kp}, we
 can show inductively that \eqref{dphink/dtm=dphimk/dtn:kp} is true for
 arbitrary $k\geqq0$.
 This is the compatibility condition of the system of differential
 equations for $f_k(t)$,
\begin{equation}
    \frac{\der f_k}{\der t_n} = \varphi_{n,k}, \qquad
    n\geqq 1.
\label{dfk/dtn=phink:kp}
\end{equation}
 As $\varphi_{n,n}=1$ and $\varphi_{n,k}=0$ for $k>n$, the solution of
 equation \eqref{dfk/dtn=phink:kp} for $k\geqq1$ with the initial
 condition $f_k(t=0)=0$ has the form
\begin{equation}
    f_k(t) = t_k + (\text{a function of }t_{k+1},t_{k+2},\dotsc)
    \qquad
    (k\geqq1)
\label{fk(t)=tk+higher:kp}
\end{equation}
 and they do not depend on $x$ because of \eqref{phink:const}.

 Now we introduce a coordinate change
 \eqref{coord-change:general-zs->kp} by
\begin{equation}
    \tilde t_n := f_n(t).
\label{tildetn=fn(t):kp}
\end{equation}
 The operator $L := e^{-f_0}\tilde L e^{f_0}$ satisfies an equation
\[
 \begin{split}
    \frac{\der L}{\der t_n} &=
    e^{-f_0} \left(
     \frac{\der \tilde L}{\der t_n}
     -
     \left[\frac{\der f_0}{\der t_n},\tilde L\right]
    \right) e^{f_0}
\\
    &=
    e^{-f_0} \left(
     \left[\sum_{k=0}^n \varphi_{n,k}\tilde B_k, \tilde L \right]
     -
     \left[\frac{\der f_0}{\der t_n},\tilde L\right]
    \right) e^{f_0}
\\
    &=
    \sum_{k=1}^n \frac{\der \tilde t_k}{\der t_n}
    e^{-f_0} [\tilde B_k, \tilde L] e^{f_0}
    =
    \sum_{k=1}^n \frac{\der \tilde t_k}{\der t_n}
    [e^{-f_0} \tilde B_k e^{f_0}, L],
 \end{split}
\]
 because of \eqref{Bn->tildeBn:temp}, \eqref{tildetn=fn(t):kp},
 \eqref{dfk/dtn=phink:kp} ($k=0$) and \eqref{phink:const}. Note that
 $e^{-f_0}\tilde B_k e^{f_0} = (L^k)_+$ because
 $
    (e^{-f_0}P e^{f_0})_+
    =
    (e^{-f_0}P_+ e^{f_0})_+ + (e^{-f_0}P_- e^{f_0})_+
    =
    e^{-f_0}P_+ e^{f_0}
 $
 for any $P\in\calE$. Hence, we have
\[
    \frac{\der L}{\der t_n}
    =
    \sum_{k=1}^n \frac{\der \tilde t_k}{\der t_n} [(L^k)_+, L].
\]
 On the other hand, we have
\[
    \frac{\der L}{\der t_n}
    =
    \sum_{k=1}^n \frac{\der \tilde t_k}{\der t_n}
    \frac{\der L}{\der \tilde t_k},
\]
 as $\der \tilde t_k/\der t_n=0$ if $k>n$ by \eqref{fk(t)=tk+higher:kp}
 and \eqref{tildetn=fn(t):kp}. Thus we have obtained equations
\[
    \sum_{k=1}^n \frac{\der \tilde t_k}{\der t_n}
    \frac{\der L}{\der \tilde t_k}
    =
    \sum_{k=1}^n \frac{\der \tilde t_k}{\der t_n} [(L^k)_+, L].
\]
 By applying the inverse matrix of the Jacobian matrix $(\der \tilde
 t_n/t_m)_{n,m=1,2,\dotsc}$ to this system, we have the Lax equations
 \eqref{kp-h:general} for $L$ with time variables $(\tilde t_1,\tilde
 t_2,\dotsc)$.

 The relations \eqref{Bn->tildeBn} of $B_n$ and $\tilde B_n$ follow from
 \eqref{Bn->tildeBn:temp}.

\bigskip
 If $b_{n,n-1}=0$ for all $n$, $\tilde L$ is of the form $\tilde L =
 \der + u_{-1}\der^{-1}+\dotsb$ as we showed in
 \propref{prop:general-zs->lax:kp}. Therefore in the above proof we have
 $[B_n, \tilde L]\in\calE(-1)$ and $B_n\semicomm{n+1}\tilde
 L$. Therefore we have $\der \varphi_{n,k}=0$ for $0\leqq k\leqq n$
 instead of $1\leqq k\leqq n$ in \eqref{phink:const}. Hence $f_0$ does
 not depend on $x$ as other $f_k$'s ($k\geqq1$) and
 $
    L=e^{-f_0}\tilde L e^{f_0} = \tilde L
 $,
 which is of the form \eqref{lax-op:kp} of the Lax operator of the KP
 hierarchy. Therefore $L$ constructed above gives a solution of the KP
 hierarchy.

 This completes the proof of \thmref{thm:general-zs->lax:kp}.
\end{proof}

\begin{rem}
\label{rem:general-zs->lax:kp:another-proof}
 This theorem can be proved using the wave operator $\hat
 W=1+w_1\der^{-1}+w_2\der^{-2}+\dotsb$ which expresses $L$ as $L=\hat
 W\der\hat W^{-1}$. The proof is parallel to that of
 \thmref{thm:general-zs->lax:mkp}.
\end{rem}

\section{Zakharov-Shabat equations for the modified KP hierarchy}
\label{sec:general-zs->lax:mKP}

\subsection{Review of the mKP hierarchy}
\label{subsec:review:mkp}

The modified KP hierarchy (or the mKP hierarchy for short) has two
formulations; one is by microdifferential operators (a series of the KP
hierarchies with compatibility conditions) and the other is by
$\Integer\times\Integer$ matrices. They are equivalent as shown by
\cite{dic:99}. Here we consider the {\em matrix mKP
hierarchy\/}.\footnote{Dickey calls the matrix formulation ``the
discrete KP hierarchy''. But usually the name ``discrete KP'' is used
for a system with discrete independent variables. So here we call the
system the matrix mKP hierarchy.}

Let us take a $\Integer\times\Integer$ matrix $\boldL$ of the form
\begin{equation}
    \boldL =
    \sum_{j=0}^\infty \diag_s[u_{1-j}(s)] \Lambda^{1-j},
\label{def:L:mkp}
\end{equation}
whose elements $u_j(s)$ belong to $\calC:=\calA[[t_1,t_2,\dotsc]]$. Here
$\Lambda$ is the shift matrix $\Lambda := (\delta_{\mu+1,\nu})_{\mu,\nu
\in \Integer}$ and $\diag_s[a(s)]$ is the diagonal matrix with the
$s$-th diagonal element $a(s)$,
\begin{equation}
    \diag_s[a(s)]:=
    \begin{pmatrix}
    \ddots &       &       &      &      &      &       \\
           & a(-2) &       &      &      &      &       \\
           &       & a(-1) &      &      &      &       \\
           &       &       & a(0) &      &      &       \\
           &       &       &      & a(1) &      &       \\
           &       &       &      &      & a(2) &       \\
           &       &       &      &      &      & \ddots \\
    \end{pmatrix}.
\label{def:diag}
\end{equation}
The following commutation relation of the matrix $\Lambda$ and a
diagonal matrix is fundamental in the following:
\begin{equation}
    \Lambda \diag_s[a(s)] = \diag_s[a(s+1)] \Lambda.
\label{Lambda=shift}
\end{equation}
In other words, we can interpret $\Lambda$ as a difference operator
$\Lambda=e^{\der_s}$, $e^{\der_s}\bigl(a(s)\bigr)=a(s+1)$. Taking this
interpretation implicitly, hereafter we do not write ``$\diag_s$''.

\medskip
We define $\boldB_n$ ($n=1,2,3,\dotsc$) from $\boldL$ by
\begin{equation}
    \boldB_n := (\boldL^n)_{+},
\label{def:boldB:mkp:matrix} 
\end{equation}
where the suffix $+$ stands for the ``upper triangular part'': for a
$\Integer\times\Integer$ matrix
$A=\sum_{s\in\Integer}a_j(s) \Lambda^j$, we denote the upper
triangular and the strictly lower triangular part by
\begin{equation}
    A_+ := \sum_{j\geq 0}a_j(s) \Lambda^j, \quad
    A_- := \sum_{j< 0}a_j(s) \Lambda^j.
\label{def:upper/lower-part}
\end{equation} 

\begin{defn}
\label{defn:mkp-h:matrix}
 The {\em matrix representation of the modified KP hierarchy}, or the
 {\em matrix mKP hierarchy} for short, is the Lax equations for
 $\boldL$,
\begin{equation}
    \frac{\der\boldL}{\der t_n} = [\boldB_n, \boldL], \qquad
    (n=1,2,3,\dotsc)
\label{mkp:matrix}
\end{equation}
 for a matrix function $\boldL$ of the form \eqref{def:L:mkp}.
\end{defn}

The Zakharov-Shabat equations for this system is
\begin{equation}
    \left[
     \frac{\der}{\der t_m}-\boldB_m, \frac{\der}{\der t_n}-\boldB_n
    \right]
    =
    \frac{\der \boldB_m}{\der t_n} - \frac{\der \boldB_n}{\der t_m}
    +[\boldB_m,\boldB_n]
    = 0, \quad
    m,n=1,2,\dotsc.
\label{zs:mkp}
\end{equation}
The whole system of these equations \eqref{zs:mkp} is equivalent to the
Lax representation \eqref{mkp:matrix}, which can be proved exactly in
the same way as the corresponding statement for the KP hierarchy in
\secref{subsec:review:kp}.

\subsection{Construction of the Lax operator of the mKP hierarchy}
\label{subsec:construction-L:mkp}

In this section we prove that the Zakharov-Shabat equations of certain
infinite matrices imply the existence of a matrix $\boldL$, which
satisfies the Lax equations.

The main strategy of the construction of $\boldL$ is similar to that for
the KP hierarchy in \secref{sec:general-zs->lax:kp}, but, since the
commutation relation \eqref{Lambda=shift} replaces the commutation
relation $[\der, f]=f'$ for the (micro)differential operators, we cannot
use the estimate \eqref{ord[A,B]} of orders.

\bigskip
We start from the $B$-operators, 
\begin{equation}
    \boldB_n 
    = \sum_{j=0}^n b^{(n)}_j(s) \Lambda^j
    \qquad
    (b^{(n)}_n(s) = 1),
\label{general-B-op:mkp}
\end{equation}
without a Lax operator. The elements $b^{(n)}_j(s)$ depend on
$t=(t_1,t_2,\dotsc)$ besides the discrete variable $s\in\Integer$.

First we prove the existence of the $\boldL$-operator satisfying the Lax
equations.

\begin{prop}
\label{prop:general-zs->lax:mkp}
 Let $\{\boldB_n\}_{n=1,2,\dotsc}$ be a sequence of matrices of the form
 \eqref{general-B-op:mkp}. If this sequence satisfies the system of
 Zakharov-Shabat equations,
\begin{equation}
    \left[
     \frac{\der}{\der t_m}-\boldB_m, \frac{\der}{\der t_n}-\boldB_n
    \right]
    =
    \frac{\der \boldB_m}{\der t_n} - \frac{\der \boldB_n}{\der t_m}
    +[\boldB_m,\boldB_n]
    = 0,
\label{general-zs:mkp}
\end{equation}
 for $m,n = 1,2,\dotsc$, then there exists a matrix of the form,
\begin{equation}
    \boldL 
    = \sum_{j=0}^\infty u_{1-j}(s) \Lambda^{1-j},
    \qquad  (u_1(s) = 1),
\label{general-L:mkp}
\end{equation}
 which satisfies the Lax equations
\begin{equation}
    \left[\frac{\der}{\der t_n}-\boldB_n,\boldL\right]
    =
    \frac{\der \boldL}{\der t_n} - [\boldB_n,\boldL]
    = 0.
\label{general-lax:mkp}
\end{equation}
 for $m,n=1,2,\dotsc$.
\end{prop}

As in the case of the KP hierarchy in the previous section, we construct
$L$-operators inductively. To measure the degree of approximation, we
often use the following sets of infinite matrices:
\begin{equation}
    \Mat^{n}
    := 
    \left\{\left.
      A = \sum_{k\leq n} a_k(s) \Lambda^k 
      \,\right|\, a_k(s)\in\calR
    \right\}.
\label{lower-matrices}
\end{equation}

\bigskip
The first approximation of $\boldL$ is obtained by the following lemma.

\begin{lem}
\label{lem:construction-L0:mkp}
 Under the assumption of \propref{prop:general-zs->lax:mkp} there exists
 $\boldL_0$ of the form \eqref{general-L:mkp} such that 
\begin{equation}
    [\boldB_n,\boldL_0] \in \Mat^0,
\label{[Bn,L0]inMat0}
\end{equation}
\end{lem}

\begin{proof}
 We construct a matrix $\boldL_0$ of the form \eqref{general-L:mkp}
 satisfying \eqref{[Bn,L0]inMat0} inductively. 

 Suppose that we have $\boldL_{0,r}$ of the form \eqref{general-L:mkp}
 satisfying
\begin{equation}
    [\boldB_n,\boldL_{0,r}]\in\Mat^0
\label{[Bn,L0r]inMat0}
\end{equation}
 for $n=1,2,\dotsc,r$. It is trivial that we can take $\boldB_1$ as
 $\boldL_{0,1}$. 

 Note that we can always expand $\boldB_n$ as 
\begin{equation}
    \boldB_n 
    =
    \sum_{k\leq n} a^{(r)}_{n,k}(s) \boldL_{0,r}^k, \qquad
    a^{(r)}_{n,n}(s)=1,
\label{Bn=series-L0r}
\end{equation}
 because of the form of $\boldL_{0,r}$, \eqref{general-L:mkp}. The
 condition \eqref{[Bn,L0r]inMat0} means that the coefficients
 $a^{(r)}_{n,k}(s)$ ($k=0,\dotsc,n$) do not depend on $s$. In fact, this
 follows from the equivalence of $[a(s),\Lambda]=0$ and $a(s+1)=a(s)$
 for all $s$.

 Therefore under our induction assumption $a^{(r)}_{n,k}(s)$
 ($n=1,\dotsc,r$, $k=0,\dotsc,n$) do not depend on $s$. Moreover, the
 coefficients $a^{(r)}_{r+1,k}(s)$ of the expansion
 \eqref{Bn=series-L0r} of $\boldB_{r+1}$ do not depend on $s$ for
 $k=3,\dotsc,r+1$. This can be shown from \eqref{general-zs:mkp} as
 follows. Using the expansion \eqref{Bn=series-L0r}, we have
\begin{equation}
 \begin{split}
    [\boldB_m,\boldB_{r+1}]
    &\equiv
    \sum_{l\leq r+1}
    \sum_{k=0}^m
      a^{(r)}_{m,k} [\boldL_{0,r}^k, a^{(r)}_{r+1,l}(s)] \boldL_{0,r}^l
    \mod\Mat^r
\\
    &=
    \sum_{l\leq r}
    \sum_{k=1}^m
     a^{(r)}_{m,k} [\boldL_{0,r}^k, a^{(r)}_{r+1,l}(s)] \boldL_{0,r}^l.
 \end{split}
\label{[Bm,Br+1]}
\end{equation}
 for $m=1,\dotsc,r$. (Note that a term
$
    a^{(r)}_{m,k} [\boldL_{0,r}^k, a^{(r)}_{r+1,l}(s)] \boldL_{0,r}^l
$
 with $k<0$ and $l\leqq r+1$ belongs to $\Mat^r$ and that
$a^{(r)}_{r+1,r+1}(s)$ does not depend on $s$, which implies that the
commutator $[\boldL_{0,r}^k,a^{(r)}_{r+1,r+1}]$ always vanishes.)  The
expression \eqref{[Bm,Br+1]} should be zero modulo $\Mat^r$, because the
Zakharov-Shabat equation \eqref{general-zs:mkp} implies
 $
    [\boldB_m,\boldB_{r+1}]
    =
    \der_{t_m}\boldB_{r+1}-\der_{t_{r+1}}\boldB_m\in\Mat^{r}
 $.

 The top degree term in the expansion of \eqref{[Bm,Br+1]} as a series
 of $\Lambda$ would come from the term with $(k,l)=(m,r)$ (if it is not
 zero), which is equal to
 $
   [\boldL_{0,r}^m, a^{(r)}_{r+1,r}(s)] \boldL_{0,r}^{r}
   =
   (a^{(r)}_{r+1,r}(s+m) - a^{(r)}_{r+1,r}(s))\Lambda^{r+m}+\dotsb
$.
 Hence, if we take $m=1$, the condition
 $[\boldB_m,\boldB_{r+1}]\in\Mat^{r}$ implies
 $a^{(r)}_{r+1,r}(s+1)=a^{(r)}_{r+1,r}(s)$, namely, $a^{(r)}_{r+1,r}(s)$
 does not depend on $s$.

 Therefore there is no contribution from the term $l=r$ in
 \eqref{[Bm,Br+1]} and the top degree term in the $\Lambda$-expansion
 would come from the term with $(k,l)=(m,r-1)$ (if it is not zero),
 which is equal to
 $
   [\boldL_{0,r}^m, a^{(r)}_{r+1,r-1}(s)] \boldL_{0,r}^{r-1}
   =
   (a^{(r)}_{r+1,r-1}(s+m) - a^{(r)}_{r+1,r-1}(s))\Lambda^{r-1+m}
   + \dotsb
 $.
 Hence, if $2\leqq m \leqq r$, the condition
 $[\boldB_m,\boldB_{r+1}]\in\Mat^{r}$ implies
 $a^{(r)}_{r+1,r-1}(s+m)=a^{(r)}_{r+1,r-1}(s)$. Moreover, if $r\geqq3$,
 $a^{(r)}_{r+1,r-1}(s+3)=a^{(r)}_{r+1,r-1}(s+2)=a^{(r)}_{r+1,r-1}(s)$,
 which means that $a^{(r)}_{r+1,r-1}(s)$ does not depend on $s$.

 Repeating the same argument inductively on $l$ downwards, we can show
 the following for $r\geqq1$.
\begin{itemize}
 \item $a^{(r)}_{r+1,r}$ is independent of $s$.
 \item If $r\geqq3$, $a^{(r)}_{r+1,r-1}$, ..., $a^{(r)}_{r+1,2}$ are
       independent of $s$.
 \item If $r\geqq2$, $a^{(r)}_{r+1,1}(s)$ is periodic:
       $a^{(r)}_{r+1,1}(s+r)=a^{(r)}_{r+1,1}(s)$. (If $r=1$, this
       holds because $a^{(r)}_{r+1,r}$ is independent of $s$.)
\end{itemize}

 Now we define $\boldL_{0,r+1}$ by modifying $\boldL_{0,r}$:
\begin{equation}
    \boldL_{0,r+1}
    :=
    \boldL_{0,r}
     + a^{(r)}_{r+1,1}(s) \boldL_{0,r}^{1-r} + b(s) \boldL_{0,r}^{-r},
\label{L0r+1=L0r+lower}
\end{equation}
 where $b(s)$ will be determined later. It is easy to see that
 $\boldL_{0,r+1}^k\equiv \boldL_{0,r}^k \mod\Mat^{-1}$ for $r\geqq2$,
 $k<r$. The $r$-th power of $\boldL_{0,r+1}$ is
\[
    \boldL_{0,r+1}^r
    \equiv
    \boldL_{0,r}^r + c_1^{(r)} \mod\Mat^{-1},
\]
 where $c_1^{(r)} = \sum_{k=0}^{r-1} a^{(r)}_{r+1,1}(s+k)$, which does
 not depend on $s$ because of the periodicity of
 $a^{(r)}_{r+1,1}(s)$. The $(r+1)$-st power is
\begin{multline*}
    \boldL_{0,r+1}^{r+1}
    \equiv
    \boldL_{0,r}^{r+1} + (c_1^{(r)}+a^{(r)}_{r+1,1}(s)) \boldL_{0,r}
\\
    +
    \left(
     (\text{terms determined by }\boldL_{0,r})
     +\sum_{k=0}^{r}b(s+k)
    \right)
    \mod\Mat^{-1}.
\end{multline*}
 Hence, if we choose $b(s)$ appropriately, we can make the last term
 $(\dotsb)$ in the above expression equal to an arbitrary sequence
 $a(s)$. (For example, first set $b(0)=b(1)=\dotsb=b(r-1)=0$ and
 determine $b(r), b(r+1), \dotsc$ recursively from $a(0), a(1),
 \dotsc$. The coefficients $b(-1), b(-2), \dotsc$ are determined from
 $a(-1), a(-2), \dotsc$.)

 Summarising, by appropriate choices of $b(s)$ in
 \eqref{L0r+1=L0r+lower} we can construct $\boldL_{0,r+1}$ whose powers
 are:
\begin{equation*}
     \boldL_{0,r+1}^k 
     \equiv
     \begin{cases}
     \boldL_{0,r}^k, & k<r, \\
     \boldL_{0,r}^r + c_1^{(r)}, &k=r, \\
     \boldL_{0,r}^{r+1}
      + (c_1^{(r)}+a^{(r)}_{r+1,1}(s))\boldL_{0,r}
      + a^{(r)}_{r+1,0}(s),
     &k=r+1,
     \end{cases}
\end{equation*}
 modulo $\Mat^{-1}$. Reversing these relations, we can rewrite
 powers of $\boldL_{0,r}$ by those of $\boldL_{0,r+1}$:
\begin{equation}
     \boldL_{0,r}^k 
     \equiv
     \begin{cases}
     \boldL_{0,r+1}^k, & k<r, \\
     \boldL_{0,r+1}^r - c_1^{(r)}, &k=r, \\
     \boldL_{0,r+1}^{r+1}
      - (c_1^{(r)}+a^{(r)}_{r+1,1}(s))\boldL_{0,r+1}
      - a^{(r)}_{r+1,0}(s),
     &r\geqq2, k=r+1, \\
     \boldL_{0,2}^2 - (c_1^{(1)}+a^{(1)}_{2,1}(s))\boldL_{0,2}
      + 2(c_1^{(1)})^2-a^{(1)}_{2,0}(s),
     &r=1, k=2,
     \end{cases}
\label{powers-of-L0r}
\end{equation}
 modulo $\Mat^{-1}$. Substituting them in the expansion
 \eqref{Bn=series-L0r} of $B_n$ ($n=1,\dotsc,r,r+1$), we obtain the
 expansion 
\begin{equation*}
    \boldB_n 
    =
    \sum_{k\leq n} a^{(r+1)}_{n,k}(s) \boldL_{0,r+1}^k, \qquad
    a^{(r+1)}_{n,n}(s)=1,
\end{equation*}
 where the coefficients $a^{(r+1)}_{n,k}(s)$ do not depend on $s$ if
 $n=1,\dotsc,r,r+1$ and $k=0,\dotsc,n$, which is the desired property of
 $\boldL_{0,r+1}$. 

 By construction \eqref{L0r+1=L0r+lower} the sequence
 $\{\boldL_{0,r}\}_{r=1,2,\dotsc}$ converges to a matrix $\boldL_0$ of
 the form \eqref{general-L:mkp} with respect to the topology induced by
 filtration $\{\Mat^r\}_{r\in\Integer}$. The limit $\boldL_0$ satisfies
 \eqref{[Bn,L0]inMat0}. In other words, $\boldB_n$'s are expanded as
\begin{equation}
    \boldB_n 
    =
    \sum_{k\leq n} a^{[0]}_{n,k}(s) \boldL_{0}^k, \qquad
    a^{[0]}_{n,n}(s)=1,
\label{Bn=series-L0}
\end{equation}
 where the coefficients $a^{[0]}_{n,k}(s)$ ($k=0,\dotsc,n$) do not
 depend on $s$. Construction of $\boldL_0$ is over.
\end{proof}

\lemref{lem:construction-L0:mkp} is the initial step of the induction in
the following lemma.

\begin{lem}
\label{lem:construction-Lr:mkp}
 Take any $r=0,1,\dotsc$. Let $\boldL_r$ be a matrix of the form
 \eqref{general-L:mkp} satisfying
\begin{equation}
    [\der_{t_n}-\boldB_n, \boldL_r]
    \in \Mat^{-r}.
\label{general-lax-r:mkp}
\end{equation}
 Then there exists $\boldL_{r+1}$ satisfying \eqref{general-lax-r:mkp}
 with $r+1$ instead of $r$. The difference of $\boldL_r$ and
 $\boldL_{r+1}$ is in $\Mat^{-r}$.
\end{lem}

\begin{proof}
 We define $g_n(s) = g_n(s;t)$ by
\begin{equation}
    [\der_{t_n} - \boldB_n, \boldL_r] \equiv g_n(s)\boldL_r^{-r}
    \mod \Mat^{-r-1}.
\label{general-laxr:def:gn}
\end{equation}
 First we prove that $g_n(s)$ does not depend on $s$. Because of the
 Zakharov-Shabat equations \eqref{general-zs:mkp},
\begin{equation*}
    [\der_{t_n} - \boldB_n, [\der_{t_m} - \boldB_m, \boldL_r]]
    =
    [\der_{t_m} - \boldB_m, [\der_{t_n} - \boldB_n, \boldL_r]].
\end{equation*}
 Hence, if $1\leqq n < m$,
\begin{equation*}
    [\der_{t_n} - \boldB_n, g_m(s)\boldL_r^{-r}]
    \equiv
    [\der_{t_m} - \boldB_m, g_n(s)\boldL_r^{-r}]
    \mod \Mat^{m-r-1},
\end{equation*}
 from which it follows that
\begin{equation*}
    0 \equiv [\boldB_m,g_n(s)\boldL_r^{-r}] \mod\Mat^{m-r-1}.
\end{equation*}
 Since the commutator in the right-hand side is expanded as
 $(g_n(s+m)-g_n(s))\Lambda^{m-r}+\dotsb$, we obtain $g_n(s+m)=g_n(s)$
 for any $m$ greater than $n$. Therefore $g_n(s)=g_n$ does not depend on
 $s$.

 Note that the assumption \eqref{general-lax-r:mkp} implies
 $[\boldB_n,\boldL_r]\in\Mat^0$, which means that the coefficients
 $a^{[r]}_{n,k}(s)$ in the expansion of $\boldB_n$,
\begin{equation}
    \boldB_n 
    =
    \sum_{k\leq n} a^{[r]}_{n,k}(s) \boldL_{r}^k, \qquad
    a^{[r]}_{n,n}(s)=1,
\label{Bn=series-Lr}
\end{equation}
 do not depend on $s$ if $k=0,\dotsc,n$, as in the proof of
 \lemref{lem:construction-L0:mkp}. 

 Substituting the expansion \eqref{Bn=series-Lr} into the
 Zakharov-Shabat equation \eqref{general-zs:mkp}, we have the following
 due to \eqref{general-laxr:def:gn}:
\begin{equation*}
 \begin{split}
    \der_{t_n}\boldB_m &=
    [\der_{t_m}-\boldB_m, \boldB_n]
\\
    &=
    \sum_{k=0}^n \der_{t_m} a^{[r]}_{n,k} \boldL_r^k
    +
    \sum_{k< 0} [\der_{t_m}-\boldB_m, a^{[r]}_{n,k}(s)] \boldL_r^k
\\
    &\ \ +
    \sum_{k\leq n}a^{[r]}_{n,k}(s)[\der_{t_m}-\boldB_m, \boldL_r^k]
\\
    &\equiv
    \sum_{k=0}^n \der_{t_m} a^{[r]}_{n,k} \boldL_r^k
    +
    \sum_{k\leq n}
    a^{[r]}_{n,k}(s)\, (k g_m \Lambda^{k-1-r}+\dotsb)
    \mod\Mat^{m-1}.
 \end{split}
\end{equation*}
 Since $\der_{t_n}\boldB_m\in\Mat^{m-1}$, we have for each
 $m=1,2,\dotsc$ and $n>m+r$ (i.e., $n-1-r>m-1$),
\begin{equation}
    \der_{t_m} a^{[r]}_{n,k}
    =
    \begin{cases}
     0 & (n-r\leqq k\leqq n),\\
     -ng_m & (k=n-1-r).
    \end{cases}
\label{da[r]n,k/dtm}
\end{equation}
 The equations \eqref{da[r]n,k/dtm} imply that the function
\begin{equation}
    f_n(t)
    :=
    -\frac{1}{n} (a^{[r]}_{n,n-1-r}(t)-a^{[r]}_{n,n-1-r}(t=0))
\label{correction-term:mkp:fn}
\end{equation}
 satisfies
\[
    \der_{t_m} f_n(t) = g_m(t),\qquad
    f_n(t=0)=0,
\]
 if $n>m+r$.

 As in \secref{sec:general-zs->lax:kp}, where we take the limit of $f_m$
 \eqref{correction-term:kp:fm}, we can take the limit
 $f:=\lim_{n\to\infty}f_n$ in the topology of $\Comp[[t]]$ defined
 by valuation, $\val t_n = n$. This limit $f$ satisfies
\begin{equation}
    \der_{t_m} f(t) = g_m(t), \qquad    f(t=0)=0. 
\label{dtmf=gm:mkp}
\end{equation}
 The matrix $\boldL_{r+1}$ satisfying \eqref{general-lax-r:mkp} for
 $r+1$ instead of $r$ is obtained by
 $\boldL_{r+1}:=\boldL_r-f\boldL_r^{-r}$. In fact,
\[
 \begin{split}
    [\der_{t_n}-\boldB_n, \boldL_{r+1}]
    &=
    [\der_{t_n}-\boldB_n, \boldL_r] - (\der_{t_n} f) \boldL_r^{-r}
    - f [\der_{t_n}-\boldB_n,\boldL_r^{-r}]
\\
    &\equiv
    f [\der_{t_n}-\boldB_n, \boldL_r^{-r}]
    \mod \Mat^{-r-1},
 \end{split}
\]
because of \eqref{general-laxr:def:gn} and \eqref{dtmf=gm:mkp}. Since
\[
    [\der_{t_n}-\boldB_n, \boldL_r^{r}]
    =
    \sum_{j=0}^r
    \boldL_r^{j} [\der_{t_n}-\boldB_n, \boldL_r] \boldL_r^{r-j-1}
    \in \Mat^{-1},
\]
 we have
$
    [\der_{t_n}-\boldB_n,\boldL_r^{-r}]
    =-\boldL_r^{-r}[\der_{t_n}-\boldB_n, \boldL_r^{r}]\boldL_r^{-r}
    \in\Mat^{-2r-1}\subset\Mat^{-r-1}
$.
 Hence, $[\der_{t_n}-\boldB_n, \boldL_{r+1}]\in\Mat^{-r-1}$. Thus the
 induction has been completed. 

 Moreover $\boldL_{r}-\boldL_{r+1}=f \boldL_r\in\Mat^{-r}$ as desired.
\end{proof}

Taking the limit $\boldL:=\lim_{r\to\infty}\boldL_r$, we obtain the
$\boldL$-operator of the form \eqref{general-L:mkp} satisfying the Lax
equations \eqref{general-lax:mkp}. This completes the proof of
\propref{prop:general-zs->lax:mkp}.
\qed

\subsection{Recovery of the mKP hierarchy}
\label{subsec:recovery-mkp}

Now we can obtain $\boldB$-operators of the mKP hierarchy from the given
$\boldB$-operators by the coordinate change.

\begin{thm}
\label{thm:general-zs->lax:mkp}
 Let $\{\boldB_n\}_{n=1,2,\dotsb}$ be a sequence of matrices in
 \propref{prop:general-zs->lax:mkp} and $\boldL$ be the matrix obtained
 there. 
 
 Then there is a coordinate change,
\begin{equation}
    t_n \mapsto \tilde t_n = t_n + O(t_{n+1},t_{n+2},\dotsc)
\label{coord-change:general-zs->mkp}
\end{equation}
 and a function $f_0=f_0(t_1,t_2,\dotsc)$ such that:
\begin{itemize}
 \item The operator $L$ satisfies Lax equations,
\begin{equation}
    \frac{\der \boldL}{\der \tilde t_n}
    =
    [\tilde \boldB_n, \boldL],
    \qquad
    \tilde \boldB_n:=(\boldL^n)_+,
\label{mkp-h:general}
\end{equation}
 with respect to the new time variables $(\tilde t_n)_{n=1,2,\dotsc}$.  
 Therefore $\boldL$ is a Lax operator of the matrix mKP hierarchy
 with respect to the variables $(\tilde t_n)_{n=1,2,\dotsc}$.

 \item Each $\boldB_n$ is expressed in terms of $\boldL$ as
\begin{equation}
    \boldB_n
    =
    \tilde \boldB_n
    + \frac{\der\tilde t_{n-1}}{\der t_n} \tilde \boldB_{n-1}
    + \dotsb +
    \frac{\der\tilde t_1}{\der t_n} \tilde \boldB_1 +
    \frac{\der f_0}{\der t_n}.
\label{boldBn->tildeboldBn}
\end{equation}
\end{itemize}
\end{thm}

\begin{proof}
 We can prove the theorem in a similar way as the case of the KP
 hierarchy in \secref{sec:general-zs->lax:kp}. Here we use another
 method. 

 It is easy to show that there exists a lower triangular matrix
 $\boldWhat$,
\begin{equation}
    \boldWhat(t)
    =
    \sum_{j=0}^\infty w_j(s) \Lambda^{-j}, \qquad w_j(s)=w_j(s;t),\ 
    w_0(s;t)=1,
\label{W:general-zs-mkp}
\end{equation}
 which satisfies
\begin{equation}
    \boldL = \boldWhat \Lambda \boldWhat^{-1}.
\label{general-zs:L=WLambdaW-1:mkp}
\end{equation}
 In fact, we have only to compare the coefficients of $\Lambda^{-j}$
 ($j=1,2,\dotsc$) in the equation
\begin{equation*}
    \boldL\boldWhat=\boldWhat\Lambda,
    \text{ i.e., }
    \sum_{j=0}^\infty
     \left(\sum_{k=0}^j u_{1-k}(s)\, w_{j-k}(s+1-k)\right)
    \Lambda^{1-j}
    =
    \sum_{j=0}^\infty w_j(s) \Lambda^{1-j},
\end{equation*}
 and determine $w_j(s)$ recursively. (The solution is not unique.) 

 Taking the adjoint of the Lax equations \eqref{general-lax:mkp} by
 $\boldWhat^{-1}$, we have
\begin{equation}
    \left[
     \frac{\der}{\der t_n} - \boldB_n^W, \Lambda
    \right]
    =0,
\label{[d-BnW,Lambda]=0:mkp}
\end{equation}
 where
\begin{equation}
    \boldB_n^W
    =
    \sum_{j=0}^\infty b^{(n),W}_{n-j}(s) \Lambda^{n-j}
    :=
    \boldWhat^{-1}\boldB_n\boldWhat
    -
    \boldWhat^{-1}\frac{\der \boldWhat}{\der t_n}.
\label{def:BnW:mkp}
\end{equation}
 The equation \eqref{[d-BnW,Lambda]=0:mkp} means that
\begin{equation}
    [\boldB_n^W,\Lambda]=0,
\label{[BnW,Lambda]=0}
\end{equation}
 namely, the coefficients $b^{(n),W}_{n-j}=b^{(n),W}_{n-j}(s)$ in the
 expansion \eqref{def:BnW:mkp} do not depend on $s$. The adjoint of the
 Zakharov-Shabat equation \eqref{general-zs:mkp} by $\boldWhat^{-1}$ is
\[
    \left[
     \frac{\der}{\der t_m} - \boldB_m^W, 
     \frac{\der}{\der t_n} - \boldB_n^W
    \right] = 0,
\]
 which reduces to
\begin{equation}
    \frac{\der \boldB_m^W}{\der t_n} = \frac{\der \boldB_n^W}{\der t_m},
\label{dBmW/dtn=dBnW/dtm:general-mkp}
\end{equation}
 because of $[\boldB_m^W,\boldB_n^W]=0$ which follows from
 \eqref{[BnW,Lambda]=0} and independence of $b^{(n),W}_{n-j}$ from
 $s$. Since each $\boldB_n^W$ is expanded as \eqref{def:BnW:mkp}, the
 coefficients of $\Lambda^k$ ($k\geqq0$) in the equations
 \eqref{dBmW/dtn=dBnW/dtm:general-mkp} ($m,n=1,2,\dotsc$) form the
 compatibility condition of the system
\begin{equation}
    \frac{\der f_k}{\der t_n} = b^{(n),W}_{k}, \qquad n=1,2,\dotsc,
\label{dfk/dtn:mkp}
\end{equation}
 where $b^{(n),W}_k=0$ for $k>n$. As $b^{(n),W}_n=b^{(n)}_n=1$, the
 solution of the system \eqref{dfk/dtn:mkp} with the initial condition
 $f_k(0)=0$ has the form
\begin{equation}
    f_k(t) = t_k + O(t_{k+1},t_{k+2},\dotsc)
\label{fk(t):general-mkp}
\end{equation}
 for $k\geqq1$. Let us use this solution $f_k(t)$ as the new
 independent variable $\tilde t_k$ ($k\geqq1$). By the definition
 \eqref{def:BnW:mkp} of $\boldB_n^W$, we have
\begin{equation*}
    \boldB_n
    -
    \frac{\der \boldWhat}{\der t_n} \boldWhat^{-1}
    =
    \sum_{j=0}^\infty b^{(n),W}_{n-j}
    \boldWhat\Lambda^{n-j}\boldWhat^{-1}.
\end{equation*}
 Taking the upper half part of this equation and using
 \eqref{general-zs:L=WLambdaW-1:mkp}, we have
\begin{equation}
    \boldB_n
    =
    \sum_{k=1}^n \frac{\der \tilde t_k}{\der t_n} \tilde\boldB_k
    +\frac{\der f_0}{\der t_n}, \qquad
    \tilde \boldB_k:=(\boldL^k)_+
\label{general-Bn=sum(dtk/dtn)tildeBk}
\end{equation}
 This is the statement \eqref{boldBn->tildeboldBn} of the theorem.

 Note that, as each $\tilde t_k=f_k(t)$ has the form
 \eqref{fk(t):general-mkp}, the Jacobian matrix of the transformation
 $(t_1,t_2,\dotsc) \mapsto (\tilde t_1,\tilde t_2,\dotsc)$ is an
 infinite upper half triangular matrix with $1$ on the
 diagonal. Therefore the inverse transformation $(\tilde t_1,\tilde
 t_2,\dotsc) \mapsto (t_1,t_2,\dotsc)$ has the same type of the Jacobian
 matrix. Multiplying that Jacobian matrix to the vector equation
\[
    \left(
    \frac{\der \boldL}{\der t_1}, \frac{\der \boldL}{\der t_2}, \dotsc
    \right)
    =
    \left(
    [\boldB_1,\boldL], [\boldB_2,\boldL], \dotsc
    \right),
\]
 which is the sequence of Lax equations \eqref{general-lax:mkp} obtained
 in \propref{prop:general-zs->lax:mkp}, we obtain the Lax equation
\[
    \frac{\der \boldL}{\der \tilde t_n}=[\tilde\boldB_n,\boldL],
    \qquad
    n=1,2,\dotsc,
\]
 as a consequence of \eqref{general-Bn=sum(dtk/dtn)tildeBk}. (The term
 $\der f_0/\der t_n$ in \eqref{general-Bn=sum(dtk/dtn)tildeBk} does not
 depend on $s$. Hence, $[\boldB_n,\boldL]=\sum_{k=1}^n (\der \tilde
 t_k/\der t_n)[\tilde \boldB_k,\boldL]$.) This is the desired Lax
 equation \eqref{mkp-h:general}.

\end{proof}

\section{Concluding remarks}
\label{sec:concluding-remarks}

In the present paper we proved that the Zakharov-Shabat equations are
sufficient even without Lax operators to define the integrable
hierarchies. The Lax operators of the KP hierarchy and the mKP hierarchy
are constructed from the information of the Zakharov-Shabat equations. 

There are several problems left to be considered:
\begin{itemize}
 \item The modified KP hierarchy has another formulation by
       (micro)differential operators which is, roughly speaking, an
       infinite collection of the KP hierarchies and compatibility
       conditions among them. See \cite{dic:99} and \cite{take:02} for
       details. We would be able to rewrite our results in this
       formulation.
 \item The problem for the Toda lattice hierarchy \cite{uen-tak:84} was
       considered in \cite{take:92} and it was claimed that the positive
       answer corresponding to \thmref{thm:general-zs->lax:kp} or
       \thmref{thm:general-zs->lax:mkp} was proved. But several serious
       gaps have been found in the proof. Especially non-monic operators
       in the Toda lattice hierarchy make the things complicated. We
       tried to reduce the problem to our results for the mKP hierarchy
       by gauge transformations (cf.\ \cite{take:90}, Proposition 1.3),
       but they cannot simplify the things completely because of the
       changes of independent variables.
 \item Similar statements are expected for other integrable hierarchies
       like the BKP hierarchy, reductions of the KP hierarchy or the
       Toda lattice hierarchy and dispersionless type hierarchies.
\end{itemize}
It is desirable to clarify these problems in further study.

\section*{Acknowledgements}
We are grateful to I. M. Krichever who suggested this problem of finding
the Lax representation from the Zakharov-Shabat equations to the second
author (T.T.) in 1991. The authors are also grateful to those who gave
valuable advices, in particular, K. Takasaki and T. Shiota.

T.T. published a part of the contents of the present paper in
\cite{take:92}, following the idea for the KP case of the first author
(M.N.), who was reluctant to be a coauthor at that time. Preparing a
talk at the memorial conference dedicated to M.N., T.T. found numerous
typos and gaps in the proofs in \cite{take:92}, so he decided to revise
the article. As M.N. once expressed regret, saying that maybe he should
have become a coauthor of \cite{take:92}, T.T. decided to publish this
paper in this form. T.T. thanks M.N.'s family for permission of
this publication.

T.T. is supported by BJNSF International Scientist Program Fund.


\end{document}